\documentclass[11pt]{article}
\usepackage{mathrsfs}
\usepackage{amsmath}
\usepackage{amsfonts}
\usepackage{amssymb}
\usepackage{latexsym}
\usepackage[dvips]{graphicx}
\newtheorem{theorem}{\bf Theorem}[section]
\newtheorem{lemma}[theorem]{\bf Lemma}

\newenvironment{proof}{\noindent{\em Proof:}}{\quad \hfill$\Box$\vspace{2ex}}

\def \bN {\Bbb N}
\def \bZ {\Bbb Z}

\def \bR {\Bbb R}

\def \cA {{\cal A}}
\def \cB {{\cal B}}

\def \cF {{\cal F}}

\def \and {\, \mbox{\rm and}\, }
\def \sinc {\,{\rm sinc}\,}

\def \supp {\,{\rm supp}\,}

\makeatletter

\newcommand{\Rmnum}[1]{\expandafter\@slowromancap\romannumeral #1@}
\makeatother
\begin{document}
\title{\bf Exponential Approximation of Bandlimited Functions from Average Oversampling
 \thanks{Supported by Guangdong
Provincial Government of China through the ``Computational Science
Innovative Research Team" program.}}
\author{Haizhang Zhang\thanks{School of Mathematics and Computational
Science and Guangdong Province Key Laboratory of Computational
Science, Sun Yat-sen University, Guangzhou 510275, P. R. China. E-mail address: {\it zhhaizh2@mail.sysu.edu.cn}. Supported in part by Natural Science Foundation of China under grants 11222103 and 11101438, and by the US Army Research Office.}}
\date{}
\maketitle
\begin{abstract}
Weighted average sampling is more practical and numerically more stable than sampling at single points as in the classical Shannon sampling framework. Using the frame theory, one can completely reconstruct a bandlimited function from its suitably-chosen average sample data. When only finitely many sample data are available, truncating the complete reconstruction series with the standard dual frame results in very slow convergence. We present in this note a method of reconstructing a bandlimited function from finite average oversampling with an exponentially-decaying approximation error.


\end{abstract}

\section{Introduction}
\setcounter{equation}{0}
The objective of this note is to show that we can achieve exponentially decaying approximation error in reconstructing a bandlimited function from its average sample data.

Denote for each $\delta>0$ by $\cB_\delta$ the Paley-Wiener space of functions $f\in L^2(\bR)\cap C(\bR)$ that are bandlimited to $[-\delta,\delta]$, that is, $\supp \hat{f}\subseteq[-\delta,\delta]$. The Fourier transform in this note takes the following form
$$
\hat{f}(\xi):=\frac1{\sqrt{2\pi}}\int_\bR f(x)e^{-ix\xi}dx,\ \ \xi\in\bR.
$$
The celebrated Shannon sampling theorem \cite{Shannon,Whittaker} states that each bandlimited function can be completely reconstructed from its samplings at the Nyquist rate. For instance, it holds for each $f\in \cB_\pi$ that
$$
f(x)=\sum_{j\in\bZ}f(j)\sinc(x-j),\ \ x\in\bR,
$$
where $\sinc(x):=\sin(\pi x)/(\pi x)$.

Sampling a function $f$ at an integer point $j$ can be viewed as applying the Delta distribution to the function $f(\cdot+j)$. The Delta distribution is natural mathematically but hard to implement physically. A more practical way is to approximate the Delta distribution by an averaging function with small support around the origin. This consideration leads to the following average sampling strategy:
\begin{equation}\label{averageform}
\mu_j(f):=\int_{-\sigma/2}^{\sigma/2}f(j+x)d\nu(x),\ \ j\in\bZ,
\end{equation}
where $\sigma$ is a small positive constant and $\nu$ is a positive Borel probability measure on $[-\sigma/2,\sigma/2]$. Compared to sampling at a single point, average sampling of the above form is also numerically more stable as the variance of the noise from the sampled values can usually be reduced by the averaging process. In fact, sophisticated algorithms based on average sampling that are highly robust to sampling noises have been proposed in \cite{DaubechiesDevore}.

Various extensions of the Shannon sampling theorem have been established for average sampling \cite{Aldroubi2002,Aldroubi2005,Grochenig,SunZhou2002a,SunZhou2002}. For instance, it was proved in \cite{Grochenig} that a function $f\in \cB_\delta$ can be completely recovered from its average sample data
$$
\int_{\bR}f(x)u_j(x)dx,\ \ j\in\bZ
$$
when
$$
0<x_{j+1}-x_j\le \sigma<\frac1{\sqrt{2}\delta},\ \ j\in\bZ
$$
and $u_j\ge0$ are nontrivial functions in $L^1(\bR)$ with $\supp u_j\subseteq [x_j-\frac\sigma2,x_j+\frac\sigma2]$.

Such average sampling theorems are obtained from the general frame theory \cite{AST,CCS,CS,Duffin,Young,Zhang}. One sees that each sampling function $\mu_j$ in (\ref{averageform}) is a continuous linear functional on $\cB_\delta$. By the Riesz representation theorem, there hence exists $g_j\in \cB_\delta$ such that
$$
\mu_j(f)=\langle f,g_j\rangle_{L^2(\bR)},\ \ j\in\bZ,
$$
where $\langle\cdot,\cdot\rangle_{L^2(\bR)}$ is the standard inner product in $L^2(\bR)$. Thus, conditions assuring that $g_j$, $j\in\bZ$ form a Riesz basis or frame for $\cB_\delta$ will immediately yield a complete reconstruction formula
$$
f=\sum_{j\in\bZ}\mu_j(f)\tilde{g_j},
$$
where $\tilde{g_j}$ denotes the standard dual frame of $g_j$. A well-known iteration scheme \cite{Daubechies} can be engaged to approximately compute $f$ from $\mu_j(f)$, $j\in\bZ$ with exponentially-decaying error.

However, when only finitely many sampling data $\mu_j(f)$ are available, this reconstruction method can be very slow. For example, when $f\in\cB_\pi$ and $\mu_j(f)=f(j)$, the standard dual frame $\tilde{g_j}$ is exactly $\sinc(\cdot-j)$. In this case, it is known \cite{Jagerman} that
$$
\sup_{x\in(0,1)}\biggl|f(x)-\sum_{j=-n}^nf(j)\sinc(x-j)\biggr|=\|f\|_{L^2(\bR)}O(\frac{1}{\sqrt{n}}).
$$
In this note, we aim at providing an explicit method with exponential approximation ability in reconstructing a bandlimited function from its finite average oversampling data (\ref{averageform}). Our idea of overcoming the above difficulty is to find a fast-decaying dual frame (actually a pseudo-frame \cite{Li2004}). The approach is also connected to the approximation question of how to smoothly extend a function so that the Fourier transform of the extended function would decay at an optimal rate.

We next introduce our main result in details. We consider functions in $\cB_\delta$ with $\delta<\pi$ and the average sampling (\ref{averageform}). The probability measure $\nu$ is required to be symmetric about the origin and satisfy
\begin{equation}\label{consgammacond}
\sigma\delta<\pi.
\end{equation}
Consequently,
\begin{equation}\label{consgamma}
\gamma:=\cos(\frac{\sigma\delta}2)>0.
\end{equation}

\begin{theorem}
Let $f\in\cB_\delta$ with $\delta<\pi$. Given the average sampling (\ref{averageform}) and the condition (\ref{consgammacond}), there exists a function $\phi\in L^2(\bR)\cap C(\bR)$ such that
$$
\sup_{x\in(0,1)}\biggl|f(x)-\frac1{\sqrt{2\pi}}\sum_{j=-n}^n \mu_j(f)\phi(x-j)\biggr|\le C\|f\|_{L^2(\bR)}\frac1{n^{3/4}}\exp\left(-\frac{2(\pi-\delta)\gamma n}{e(1+\sqrt{2})^2(\gamma+\sigma(\pi-\delta))}\right),
$$
where $C$ is a positive constant independent of $f$ and $n$.
\end{theorem}

The above theorem will be proved in the next section. In particular, the crucial function $\phi$ will be explicitly constructed. Numerical experiments to justify our result are presented in Section 3.

\section{Exponential Approximation Reconstruction}
\setcounter{equation}{0}
We shall see that $\{\mu_j(f):j\in\bZ\}$ defined by (\ref{averageform}) can be represented through a frame in $\cB_\delta$. Our approach is to find a dual frame that decays fast. To this end, we recall a few basic facts about the Paley-Wiener space $\cB_\delta$.

The space $\cB_\delta$ endowed with the $L^2$-norm on $\bR$ is a reproducing kernel Hilbert space with $\sin(\delta(x-y))/(\pi(x-y))$ as its reproducing kernel. In other words, we have for each $f\in\cB_\delta$ and each $x\in\bR$
$$
f(x)=\int_\bR f(y)\frac{\sin\delta(x-y)}{\pi(x-y)}dy.
$$
It happens that $\sinc(\cdot-j)$, $j\in\bZ$ form an orthonormal basis for $\cB_\pi$, which implies the useful Parseval identity in $\cB_\pi$
\begin{equation}\label{parseval}
\|f\|_{L^2(\bR)}^2=\sum_{j\in\bZ}|\langle f,\sinc(\cdot-j)\rangle_{L^2(\bR)}|^2=\sum_{j\in\bZ}|f(j)|^2,\ \ f\in\cB_\pi.
\end{equation}
Also, $\cB_\pi$ is translation-invariant in the sense that for all $f\in\cB_\pi$ and $x\in \bR$, $f(\cdot+x)$ remains in $\cB_\pi$ and has the same norm as that of $f$.

When $\delta<\pi$, $\cB_\delta$ is a subspace of $\cB_\pi$. The identity (\ref{parseval}) hence holds true for functions $f\in\cB_\delta$. Immediately, we get for $f\in\cB_\delta$
\begin{equation}\label{frameupper}
\sum_{j\in\bZ}|\mu_j(f)|^2=\sum_{j\in\bZ}\biggl|\int_{-\sigma/2}^{\sigma/2}f(j+x)d\nu(x)\biggr|^2\le \int_{-\sigma/2}^{\sigma/2}\biggl(\sum_{j\in\bZ}|f(j+x)|^2\biggr)d\nu(x)=\|f\|_{\cB_\delta}.
\end{equation}
It follows that each sampling function $\mu_j$ is a continuous linear functional on $\cB_\delta$. Thus, by the Riesz representation theorem, there exists some $g_j\in\cB_\delta$ such that
\begin{equation}\label{gj}
\mu_j(f)=\langle f,g_j\rangle_{L^2(\bR)},\ \ f\in\cB_\delta.
\end{equation}
Inequality (\ref{frameupper}) indicates that $\{g_j:j\in\bZ\}$ is a Bessel sequence (see, \cite{Young}, page 154) for $\cB_\delta$. Furthermore, we shall see that under our conditions on $\nu$, $g_j$ in fact constitutes a frame for $\cB_\delta$. In order to obtain a method of fast reconstructing $f\in\cB_\delta$ from the finite sample data $\mu_j(f)$, $|j|\le n$, we hope to find a dual frame for $\{g_j\}$ that is fast-decaying at infinity. The following lemma finds all the dual frames for $\{g_j\}$ formed by the integer shifts of a single function.

Introduce the exponential function
$$
W(\xi):=\int_{-\sigma/2}^{\sigma/2}e^{it\xi}d\nu(t),\ \ \xi\in\bR.
$$
Since $\nu$ is symmetric about the origin, $W$ is real-valued. By the condition (\ref{consgammacond}) that $\sigma\delta<\pi$, we have
\begin{equation}\label{Wlowerbound}
0<\gamma=\cos(\frac{\sigma\delta}2)\le W(\xi)\le 1,\ \ \xi\in[-\delta,\delta].
\end{equation}

\begin{lemma}\label{dualframe}
Let $\phi\in C(\bR)\cap L^2(\bR)$ with $\supp\hat{\phi}\subseteq [-2\pi+\delta,2\pi-\delta]$. Then the identity
\begin{equation}\label{exactexpansion}
f=\frac1{\sqrt{2\pi}}\sum_{j\in\bZ}\mu_j(f)\phi(\cdot-j)
\end{equation}
holds in $L^2(\bR)$ for all $f\in\cB_\delta$ if and only if
\begin{equation}\label{exactexpansioncond}
\hat{\phi}(\xi)W(\xi)=1,\ \ \mbox{for almost every }\xi\in[-\delta,\delta].
\end{equation}
\end{lemma}
\begin{proof}
Let $f\in\cB_\delta$ and $\phi\in C(\bR)\cap L^2(\bR)$ with $\supp\hat{\phi}\subseteq [-2\pi+\delta,2\pi-\delta]$. For simplicity, denote by $g$ the right hand side of (\ref{exactexpansion}). We compute that
\begin{equation}\label{dualframeeq1}
\hat{g}(\xi)=\hat{\phi}(\xi)\frac1{\sqrt{2\pi}}\sum_{j\in\bZ}\mu_j(f)e^{-ij\xi}=\hat{\phi}(\xi)\int_{-\sigma/2}^{\sigma/2}\biggl(\frac1{\sqrt{2\pi}}\sum_{j\in\bZ}f(j+t)e^{-ij\xi}\biggr)d\nu(t),\ \ \xi\in\bR.
\end{equation}
Observe that $\frac1{\sqrt{2\pi}}\sum_{j\in\bZ}f(j+t)e^{-ij\xi}$ is the expansion of $\hat{f}e^{it\cdot}$ with respect to the orthonormal basis $\{\frac1{\sqrt{2\pi}}e^{-ij\xi}:j\in\bZ\}$ in $L^2([-\pi,\pi])$. Combing this observation with (\ref{dualframeeq1}) yields
$$
\hat{g}(\xi)=\hat{\phi}(\xi)(\hat{f}(\xi)W(\xi))_{2\pi},\ \ \xi\in\bR,
$$
where the subindex $2\pi$ denotes the $2\pi$-periodic extension of a function originally defined only within $[-\pi,\pi]$. This Fourier transform $\hat{g}$ equals $\hat{f}$ for all $f\in\cB_\delta$ if and only if (\ref{exactexpansioncond}) holds.
\end{proof}

We remark that the arguments above imply that the functions $g_j$ determined by (\ref{gj}) form a frame for $\cB_\delta$. In fact, for $f\in\cB_\delta$,
$$
\sum_{j\in\bZ}|\mu_j(f)|^2=\frac1{2\pi}\int_{-\pi}^{\pi}\biggl|\sum_{j\in\bZ}\mu_j(f)e^{-ij\xi}\biggr|^2d\xi=\int_{-\delta}^{\delta}|\hat{f}(\xi)|^2|W(\xi)|^2d\xi.
$$
This together with (\ref{Wlowerbound}) implies
\begin{equation}\label{framecond}
\gamma \|f\|_{\cB_\delta}\le\biggl(\sum_{j\in\bZ}|\mu_j(f)|^2\biggr)^{1/2}\le \|f\|_{\cB_\delta}.
\end{equation}
Therefore, $\{g_j\}$ indeed is a frame for $\cB_\delta$ with frame bounds $\gamma$ and $1$.

We shall carefully choose a function $\phi\in C(\bR)\cap L^2(\bR)$ satisfying $\supp\hat{\phi}\subseteq[-2\pi+\delta,2\pi-\delta]$ and the condition (\ref{exactexpansioncond}). Our method of reconstructing the values of a function $f\in \cB_\delta$ on $(0,1)$ from its finite sample data $\mu_j(f)$, $-n\le j\le n$, is directly given as
\begin{equation}\label{reconstruction}
(\cA_nf)(x):=\frac1{\sqrt{2\pi}}\sum_{j=-n}^n \mu_j(f)\phi(x-j),\ \ x\in(0,1).
\end{equation}
We give an initial analysis of the approximation error of this method.

\begin{lemma}\label{truncated}
Let $\phi\in C(\bR)\cap L^2(\bR)$ satisfying $\supp\hat{\phi}\subseteq[-2\pi+\delta,2\pi-\delta]$ and (\ref{exactexpansioncond}). It holds for all $f\in\cB_\delta$ and $x\in(0,1)$ that
\begin{equation}\label{truncatederror}
|f(x)-(\cA_nf)(x)|\le \frac1{\sqrt{2\pi}}\|f\|_{L^2(\bR)}\biggl(\sum_{|j|>n}|\phi(x-j)|^2\biggr)^{1/2}.
\end{equation}
\end{lemma}
\begin{proof}
Under the assumptions, (\ref{exactexpansion}) holds in $L^2(\bR)$ for all $f\in \cB_\delta$. It is straightforward to show by (\ref{framecond}) and the Parseval identity that the series on the right hand side of (\ref{exactexpansion}) converges uniformly on $\bR$. Thus, it defines a continuous function on $\bR$. Since $f\in \cB_\delta\subseteq C(\bR)$, that (\ref{exactexpansion}) holds in $L^2(\bR)$ implies that it also holds true pointwise on $\bR$. Thus, we have for all $x\in(0,1)$
$$
|f(x)-(\cA_nf)(x)|=\frac1{\sqrt{2\pi}}\biggl|\sum_{|j|>n}\mu_j(f)\phi(x-j)\biggr|.
$$
By the Cauchy-Schwartz inequality and the frame property (\ref{framecond}), we get
$$
|f(x)-(\cA_nf)(x)|\le \frac1{\sqrt{2\pi}}\biggl(\sum_{|j|>n}|\mu_j(f)|^2\biggr)^{1/2}\biggl(\sum_{|j|>n}|\phi(x-j)|^2\biggr)^{1/2}\le \frac1{\sqrt{2\pi}}\|f\|_{\cB_\delta}\biggl(\sum_{|j|>n}|\phi(x-j)|^2\biggr)^{1/2},
$$
which completes the proof.
\end{proof}

By (\ref{truncatederror}), $\phi$ should be fast-decaying at infinity. We make use of a well-known property of the Fourier transform that if $\hat{\phi}$ has $k-1$ continuous derivatives and $(\hat{\phi})^{(k-1)}$ is absolutely continuous then
\begin{equation}\label{fourierdecaying}
|\phi(x-j)|\le \frac1{\sqrt{2\pi}}\frac{\|\hat{\phi}^{(k)}\|_{L^1([-2\pi+\delta,2\pi-\delta])}}{|x-j|^k},\ \ j\ge 1,\ x\in(0,1).
\end{equation}
We shall choose a $\phi$ so that the important quantity $\|\hat{\phi}^{(k)}\|_{L^1([-2\pi+\delta,2\pi-\delta])}$ is minimized. Considering (\ref{exactexpansioncond}) and that $W$ is an even function, we would like $\hat{\phi}$ to be even as well. Thus,
\begin{equation}\label{l1normsum}
\|(\hat{\phi})^{(k)}\|_{L^1([-2\pi+\delta,2\pi-\delta])}=\left\| (\frac1{W})^{(k)}\right\|_{L^1([-\delta,\delta])}+2\|\hat{\phi}^{(k)}\|_{L^1([\delta,2\pi-\delta])}.
\end{equation}

Denote for all intervals $[a,b]$ by $\cF_k[a,b]$ the class of functions $\varphi\in C^{(k-1)}([a,b])$ with $\varphi^{(k-1)}$ being absolutely continuous. Our task is to extend the values of $1/W$ on $[-\delta,\delta]$ to an even function $\hat{\phi}\in \cF_k[-2\pi+\delta,2\pi-\delta]$ with $\supp\hat{\phi}\subseteq [-2\pi+\delta,2\pi-\delta]$ such that the $L^1$ norm of its $k$-th derivative is minimized. We formulate this minimization problem below.

Set
\begin{equation}\label{djs}
d_j:=(\frac1W)^{(j)}(\delta),\ \ 0\le j\le k-1.
\end{equation}
To extend $ 1/W$ on $[-\delta,\delta]$ to an even function $\hat{\phi}\in \cF_k[-2\pi+\delta,2\pi-\delta]$, we are looking for a function $\hat{\phi}\in \cF_k[\delta,2\pi-\delta]$ such that
\begin{equation}\label{deltajoint}
\hat{\phi}^{(j)}(\delta)=d_j,\ \hat{\phi}^{(j)}(2\pi-\delta)=0,\ \ 0\le j\le k-1.
\end{equation}
Thus, we want to solve
$$
\inf_{\hat{\phi}\in \cF_k[\delta,2\pi-\delta]}\|\hat{\phi}^{(k)}\|_{L^1([\delta,2\pi-\delta])}
$$
subject to the condition (\ref{deltajoint}). This turns out to be hard to solve due to the nature of $L^1$ norm. Since
\begin{equation}\label{cauchyschwarz12}
\|\hat{\phi}^{(k)}\|_{L^1([\delta,2\pi-\delta])}\le \sqrt{2\pi-2\delta}\|\hat{\phi}^{(k)}\|_{L^2([\delta,2\pi-\delta])}.
\end{equation}
we shall try to find
\begin{equation}\label{Vkoriginal}
V_k:=\inf_{\phi\in \tilde{\cF}_k[\delta,2\pi-\delta]}\sqrt{2\pi-2\delta}\|\hat{\phi}^{(k)}\|_{L^2([\delta,2\pi-\delta])},
\end{equation}
where $\tilde{\cF}_k[a,b]$ denotes the class of functions $\varphi\in \cF_k[a,b]$ with $\varphi^{(k)}\in L^2([a,b])$. Through a change of variables
\begin{equation}\label{changevariables}
\psi(t):=\hat{\phi}(2\pi-\delta-(2\pi-2\delta)t), \ \ t\in[0,1],
\end{equation}
we observe that
\begin{equation}\label{Vk}
V_k=\frac1{(2\pi-2\delta)^{k-1}}\inf_{\psi\in \tilde{\cF}_k[0,1]}\|\psi^{(k)}\|_{L^2([0,1])}
\end{equation}
subject to
\begin{equation}\label{deltajoint2}
\psi^{(j)}(0)=0,\ \psi^{(j)}(1)=d_j':=(-1)^j(2\pi-2\delta)^jd_j,\ \ 0\le j\le k-1.
\end{equation}

A major technical part of this section is to solve this minimization problem. To this end, we first give an integral reformulation of the restriction condition (\ref{deltajoint2}).

\begin{lemma}\label{extension}
Let $\psi\in\tilde{\cF}_k[0,1]$ with $\psi^{(j)}(0)=0$, $1\le j\le k-1$. Then it satisfies
\begin{equation}\label{deltajoint2second}
\psi^{(j)}(1)=d_j',\ \ 0\le j\le k-1
\end{equation}
if and only if
\begin{equation}\label{integralcond}
\int_0^1 \psi^{(k)}(t)t^jdt=q_j,\ \ 0\le j\le k-1,
\end{equation}
where
\begin{equation}\label{qjs}
q_j:=d_{k-1}'+\sum_{l=1}^j(-1)^l\frac{j!}{(j-l)!}d_{k-l-1}',\ \ 0\le j\le k-1.
\end{equation}
\end{lemma}
\begin{proof}
Let $\psi\in\tilde{\cF}_k[0,1]$  with $\psi^{(j)}(0)=0$, $1\le j\le k-1$. Suppose that it satisfies (\ref{deltajoint2second}). One proves (\ref{integralcond}) by induction and by integration by parts. Conversely, suppose (\ref{integralcond}) holds true. We first see
$$
\psi(t)=\int_0^td\tau_1\int_0^{\tau_1}d\tau_2\cdots\int_0^{\tau_{k-1}}\psi^{(k)}(\tau_k)d\tau_k, \ \ t\in[0,1].
$$
Thus
$$
\psi^{(k-1)}(1)=\int_0^1\psi^{(k)}(\tau_k)d\tau_k=q_0=d_{k-1}'.
$$
Suppose $\psi^{(l)}(1)=d_l'$ for $k-j+1\le l\le k-1$. We use two steps of integration by parts to get
$$
\begin{array}{ll}
\psi^{(k-j)}(1)&\displaystyle{=\int_0^1d\tau_{k-j+1}\int_0^{\tau_{k-j+1}}d\tau_{k-j+2}\cdots\int_0^{\tau_{k-1}}\psi^{(k)}(\tau_k)d\tau_k}\\
&\displaystyle{=\int_0^1d\tau_{k-j+2}\cdots\int_0^{\tau_{k-1}}\psi^{(k)}(\tau_k)d\tau_k-\int_0^1\tau_{k-j+1}\int_0^{\tau_{k-j+1}}d\tau_{k-j+2}\cdots\int_0^{\tau_{k-1}}\psi^{(k)}(\tau_k)d\tau_k}\\
&\displaystyle{=d_{k-j+1}'-\int_0^1\tau_{k-j+1}\int_0^{\tau_{k-j+1}}d\tau_{k-j+2}\cdots\int_0^{\tau_{k-1}}\psi^{(k)}(\tau_k)d\tau_k}\\
&\displaystyle{=d_{k-j+1}'-\frac12d_{k-j+2}'+\frac12\int_0^1\tau_{k-j+2}^2\int_0^{\tau_{k-j+2}}d\tau_{k-j+3}\cdots\int_0^{\tau_{k-1}}\psi^{(k)}(\tau_k)d\tau_k}.
\end{array}
$$
Successively using integration by parts, we obtain
$$
\psi^{(k-j)}(1)=d_{k-j+1}'+\sum_{l=1}^{j-2}\frac{(-1)^l}{(l+1)!}d_{k-j+1+l}'+\frac{(-1)^{j-1}}{(j-1)!}\int_0^1\tau^{j-1}\psi^{(k)}(\tau)d\tau.
$$
Substituting (\ref{integralcond}) into the above equation yields
$$
\psi^{(k-j)}(1)=d_{k-j+1}'+\sum_{l=1}^{j-2}\frac{(-1)^l}{(l+1)!}d_{k-j+1+l}'+\frac{(-1)^{j-1}}{(j-1)!}q_{j-1}.
$$
One verifies that the right hand side above does equal $d_{k-j}'$.
\end{proof}

Using Lemma \ref{extension}, we are able to solve the minimization problem (\ref{Vk}).
\begin{lemma}\label{vkhilbert}
Let $V_k$ be given by (\ref{Vkoriginal}). Then
$$
V_k=\frac{(q^TH_k^{-1}q)^{1/2}}{(2\pi-2\delta)^{k-1}},
$$
where $q=(q_j:0\le j\le k-1)^T$ and $H_k$ is the $k\times k$ Hilbert matrix
$$
H_k(i,j):=\frac1{i+j+1},\ \ 0\le i,j\le k-1.
$$
Moreover, the function $\hat{\phi}$ that attains the infimum (\ref{Vkoriginal}) is uniquely determined by
\begin{equation}\label{minimizer}
\hat{\phi}^{(k)}(\xi):=\frac{(-1)^k}{(2\pi-2\delta)^k}\sum_{j=0}^{k-1}(H_k^{-1}q)_j\left(\frac{2\pi-\delta-\xi}{2\pi-2\delta}\right)^j,\ \ \xi\in [\delta,2\pi-\delta].
\end{equation}
\end{lemma}
\begin{proof}
By equation (\ref{Vk}) and Lemma \ref{extension},
$$
V_k=\frac1{(2\pi-2\delta)^{k-1}}\inf_{g\in L^2([0,1])}\|g\|_{L^2([0,1])}
$$
subject to
\begin{equation}\label{integralcond2}
\int_0^1 g(t)t^jdt=q_j,\ \ 0\le j\le k-1.
\end{equation}
By the orthogonal decomposition in a Hilbert space, the minimizer $g$ of the above minimization problem is unique and must be of the form
$$
g(t)=\sum_{j=0}^{k-1}c_jt^j,\ \ 0\le t\le 1.
$$
The above polynomial satisfies (\ref{integralcond2}) if and only if the coefficient vector $c=(c_j:0\le j\le k-1)^T$ satisfies
$$
H_kc=q.
$$
Let $c:=H_k^{-1}q$. We compute that
$$
\biggl\|\sum_{j=0}^{k-1}c_jt^j\biggr\|^2_{L^2([0,1])}=c^TH_kc=q^TH_k^{-1}q.
$$
The only minimizer $\psi\in\tilde{\cF}_k[0,1]$ that attains (\ref{Vk}) is hence given by $\psi^{(k)}=g$. By (\ref{changevariables}), there exists a unique minimizer $\hat{\phi}$ for (\ref{Vkoriginal}), which is given by (\ref{minimizer}).
\end{proof}

Let $\hat{\phi}$ be determined by
\begin{equation}\label{optimalphi}
\hat{\phi}(\xi)=\frac1{W(\xi)},\ |\xi|\le \delta\mbox{ and }\hat{\phi}^{(k)}(\xi)=\frac{(-1)^k}{(2\pi-2\delta)^k}\sum_{j=0}^{k-1}(H_k^{-1}q)_j\left(\frac{2\pi-\delta-|\xi|}{2\pi-2\delta}\right)^j,\ \ \  \delta\le |\xi|\le 2\pi-\delta.
\end{equation}
Finally, we shall use the following well-known fact about the smallest eigenvalue $\rho_{min}(H_k)$ of the Hilbert matrix $H_k$
\begin{equation}\label{smalleigenvalue}
\left(\frac1 {\rho_{min}(H_k)}\right)^{1/2}\le \frac {C_H}{k^{1/4}}(1+\sqrt{2})^{2k},
\end{equation}
where $C_H$ is a constant independent of $k$ (see, \cite{Wilf}, page 51).

We are ready to prove the approximation error for the reconstruction method (\ref{reconstruction}). 

\begin{theorem}\label{main}
Let $\phi$ be defined by (\ref{optimalphi}) with
\begin{equation}\label{choice}
k:=1+\lfloor \frac{n}{\beta e}\rfloor,\ \
\beta:=(1+\sqrt{2})^2\frac{\gamma+\sigma(\pi-\delta)}{2\gamma(\pi-\delta)}.
\end{equation}
Then for all $f\in\cB_\delta$, $x\in(0,1)$, and $n\in\bN$ satisfying $n\ge \beta e$ and
\begin{equation}\label{kcondition}
\frac{e\sigma\delta}{\gamma}k^{3/4}\le 4(1+\sqrt{2})^{2k},
\end{equation}
the following inequality holds
\begin{equation}\label{ultimateestimate}
|f(x)-(\cA_nf)(x)|\le\|f\|_{L^2(\bR)}\frac{4(1+\sqrt{2})^2C_H}{\gamma\pi}(\frac2\beta)^{1/4}e^{3/4}\sqrt{1+2\beta e}\frac1{n^{3/4}}\exp(-\frac{n}{\beta e}).
\end{equation}
\end{theorem}
\begin{proof}
We shall use the estimate (\ref{truncatederror}). To this end, {\bf we first bound the $L^1$ norm of $\hat{\phi}^{(k)}$}. By equations (\ref{l1normsum}) and (\ref{cauchyschwarz12}), and Lemma \ref{vkhilbert},
\begin{equation}\label{phil1normbound}
\|\hat{\phi}^{(k)}\|_{L^1([-2\pi+\delta,2\pi-\delta])}\le \left\| (\frac1{W})^{(k)}\right\|_{L^1([-\delta,\delta])}+2V_k=\left\| (\frac1{W})^{(k)}\right\|_{L^1([-\delta,\delta])}+2\frac{(q^TH_k^{-1}q)^{1/2}}{(2\pi-2\delta)^{k-1}}.
\end{equation}

{\bf We start with the $L^1$ norm of $(\frac1{W})^{(k)}$}. Let
$$
p_j:=\left\| (\frac1{W})^{(j)}\right\|_{L^\infty([-\delta,\delta])},\ \ 0\le j\le k-1.
$$
Recall the two constants $\gamma$, $\sigma$ and equation (\ref{Wlowerbound}). Another simple fact to be used is
\begin{equation}\label{boundwj}
|W^{(j)}(\xi)|=\left|\int_{-\sigma/2}^{\sigma/2}(it)^je^{it\xi}d\nu(t)\right|\le \frac{\sigma^j}{2^j},\ \ 0\le j\le k-1.
\end{equation}
We shall prove by induction that
\begin{equation}\label{linfty1w}
p_j\le \frac{\sigma^j}{2^j\gamma^{j+1}}j^j,\ \ 0\le j\le k-1.
\end{equation}
Here $0^0:=1$. Clearly, this is true for $j=0$ by (\ref{Wlowerbound}). Suppose it holds true for $j\le k-1$. Set $h:=1/W$. We apply the Leibniz formula to compute the $k$-th derivatives of both sides of $hW=1$ to get
$$
h^{(k)}=-\frac1W\sum_{j=0}^{k-1}{{k}\choose{j}}h^{(j)}W^{(k-j)}.
$$
Thus, by induction on $p_j$, equations (\ref{Wlowerbound}) and (\ref{boundwj}),
$$
p_k\le \frac1{\gamma}\sum_{j=0}^{k-1}{{k}\choose{j}}\frac{\sigma^j}{2^j\gamma^{j+1}}j^j\frac{\sigma^{k-j}}{2^{k-j}}\le\frac{\sigma^k}{2^k\gamma^{k+1}}
\sum_{j=0}^{k}{{k}\choose{j}}(k-1)^j=\frac{\sigma^k}{2^k\gamma^{k+1}}k^k,
$$
which proves (\ref{linfty1w}). As a direct consequence,
\begin{equation}\label{l1norm1w}
\left\| (\frac1{W})^{(k)}\right\|_{L^1([-\delta,\delta])}\le 2\delta p_k=2\delta\frac{\sigma^k}{2^k\gamma^{k+1}}k^k.
\end{equation}

{\bf We next estimate $q^TH_k^{-1}q$}. Obviously,
$$
q^TH_k^{-1}q\le \frac{\|q\|^2}{\rho_{min}(H_k)},
$$
where $\|q\|$ is the standard Euclidean norm of $q$. By (\ref{smalleigenvalue}),
\begin{equation}\label{qnormestimate}
(q^TH_k^{-1}q)^{1/2}\le \frac{C_H}{k^{1/4}}(1+\sqrt{2})^{2k}\|q\|.
\end{equation}
We need to bound $\|q\|$. By (\ref{qjs}), (\ref{linfty1w}), and $|d_j|\le p_j$, we get for $0\le j\le k-1$,
$$
\begin{array}{ll}
q_j&\displaystyle{\le |d_{k-1}'|+\sum_{l=1}^j{{j}\choose{l}}l!|d_{k-l-1}'|}\\
&\displaystyle{\le (\pi-\delta)^{k-1}\frac{\sigma^{k-1}}{
\gamma^k}(k-1)^{k-1}+\sum_{l=1}^j{{j}\choose{l}}(\pi-\delta)^{k-l-1}\frac{\sigma^{k-l-1}}{\gamma^{k-l}}(k-l-1)^{k-l-1}l!}
\end{array}
$$
Using $(k-l-1)^{k-l-1}l!\le (k-1)^{k-2}$, we have
$$
\sum_{l=1}^j{{j}\choose{l}}(\pi-\delta)^{k-l-1}\frac{\sigma^{k-l-1}}{\gamma^{k-l}}(k-l-1)^{k-l-1}l!\le (\pi-\delta)^{k-1}\frac{\sigma^{k-1}}{
\gamma^k}(k-1)^{k-2}\biggl(1+\frac{\gamma}{\sigma(\pi-\delta)}\biggr)^j.
$$
By the above two equations, we estimate that
\begin{equation}\label{qnorm}
\|q\|\le 2\sqrt{k}(\pi-\delta)^{k-1}\frac{\sigma^{k-1}}{
\gamma^k}(k-1)^{k-1}\biggl(1+\frac{\gamma}{\sigma(\pi-\delta)}\biggr)^{k-1}.
\end{equation}

We combine (\ref{phil1normbound}), (\ref{l1norm1w}), (\ref{qnormestimate}), and (\ref{qnorm}) to get that when (\ref{kcondition}) is satisfied,
\begin{equation}\label{cruciall1norm}
\|\hat{\phi}^{(k)}\|_{L^1([-2\pi+\delta,2\pi-\delta])}\le \frac{8(1+\sqrt{2})^2C_H}{\gamma }k^{1/4}\left(\beta(k-1)\right)^{k-1}.
\end{equation}

Finally, we apply (\ref{fourierdecaying}) to (\ref{truncatederror}) to get for $x\in(0,1)$
$$
|f(x)-(\cA_nf)(x)|\le \|f\|_{L^2(\bR)}\frac{4(1+\sqrt{2})^2C_H}{\gamma \pi }k^{1/4}\left(\beta(k-1)\right)^{k-1}\left(\sum_{|j|> n}\frac1{|x-j|^{2k}}\right)^{1/2}.
$$
Notice that for $x\in(0,1)$,
$$
\sum_{|j|> n}\frac1{|x-j|^{2k}}\le \frac1{n^{2k}}+2\sum_{j=n+1}^\infty\frac1{j^{2k}}\le\frac1{n^{2k}}+2\int_{n}^\infty \frac1{t^{2k}}dt=(1+\frac{2n}{2k-1})\frac1{n^{2k}}.
$$
By the above two equations,
$$
|f(x)-(\cA_nf)(x)|\le \|f\|_{L^2(\bR)}\frac{4(1+\sqrt{2})^2C_H}{\gamma\pi}\frac{k^{1/4}}{n}\sqrt{1+\frac{2n}{2k-1}}\left(\frac{\beta(k-1)}{n}\right)^{k-1}.
$$
With the optimal choice (\ref{choice}), we reach
$$
|f(x)-(\cA_nf)(x)|\le \|f\|_{L^2(\bR)}\frac{4(1+\sqrt{2})^2C_H}{\gamma\pi}(\frac2\beta)^{1/4}e^{3/4}\sqrt{1+2\beta e}\frac1{n^{3/4}}\exp(-\frac{n}{\beta e}).
$$
The proof is complete.
\end{proof}

At the end of the section, we remark that it has been proved in \cite{MXZ,Qian} that exponentially decaying approximation error can be achieved in reconstructing a function $f\in\cB_\delta$ with $\delta<\pi$ from the oversampling data $f(j)$, $-n\le j\le n$. Two reconstruction algorithms were proposed therein. The one in \cite{Qian} uses a Gaussian regularizer and is thus completely different. The analysis in \cite{MXZ} essentially corresponds to the special case $\sigma=0$ and $W(\xi)\equiv1$ here. The discussion of general average sampling in this note is much more complicated.

\section{Numerical Experiments}
\setcounter{equation}{0}
In this section, we present two numerical experiments to illustrate our reconstruction method. In both experiments, the target function is
$$
f(x):=\frac{\sin(\delta x)}{\pi x},\ \ x\in\bR.
$$
We will compute the approximation error
\begin{equation}\label{computederr}
\max_{1\le j\le 9}\left|f(\frac j{10})-(\cA_n f)(\frac j {10})\right|.
\end{equation}
Our purpose is to show that the error does satisfy the estimate (\ref{ultimateestimate}) in Theorem \ref{main} and therefore to verify that it does decay exponentially. Our method requires solving linear equations with the Hilbert matrices as the coefficient matrix. The Hilbert matrices with large size are highly ill-conditioned. Fortunately, the Hilbert matrix involved in our method is of order $k=1+\lfloor \frac{n}{\beta e}\rfloor$, which is typically very small even for considerably large $n$. Note that $n$ is the number of sampling points. In fact, in both experiments, $k$ is not more than $5$ when the approximation error is already satisfactory.\newline

\noindent{\bf Experiment 1.} In this experiment, the averaging sampling is
$$
\mu_j(f):=\frac1{12}f(j-\frac1{8})+\frac1{12}f(j-\frac1{16})+\frac23f(j)+\frac1{12}f(j+\frac1{16})+\frac1{12}f(j+\frac1{8}),\ \ j\in\bZ,\ f\in\cB_\delta.
$$
We compute the approximation error (\ref{computederr}) and the projected error in the estimate (\ref{ultimateestimate}) for $\delta\in\{\frac{\pi}4,\frac{\pi}2,\frac{2\pi}3\}$ and for $n\in\{14,16,18,20,22,24\}$. The results are plotted below. The estimate (\ref{ultimateestimate}) is hence verified.

\begin{center}
\scalebox{0.55}[0.5]{\includegraphics*{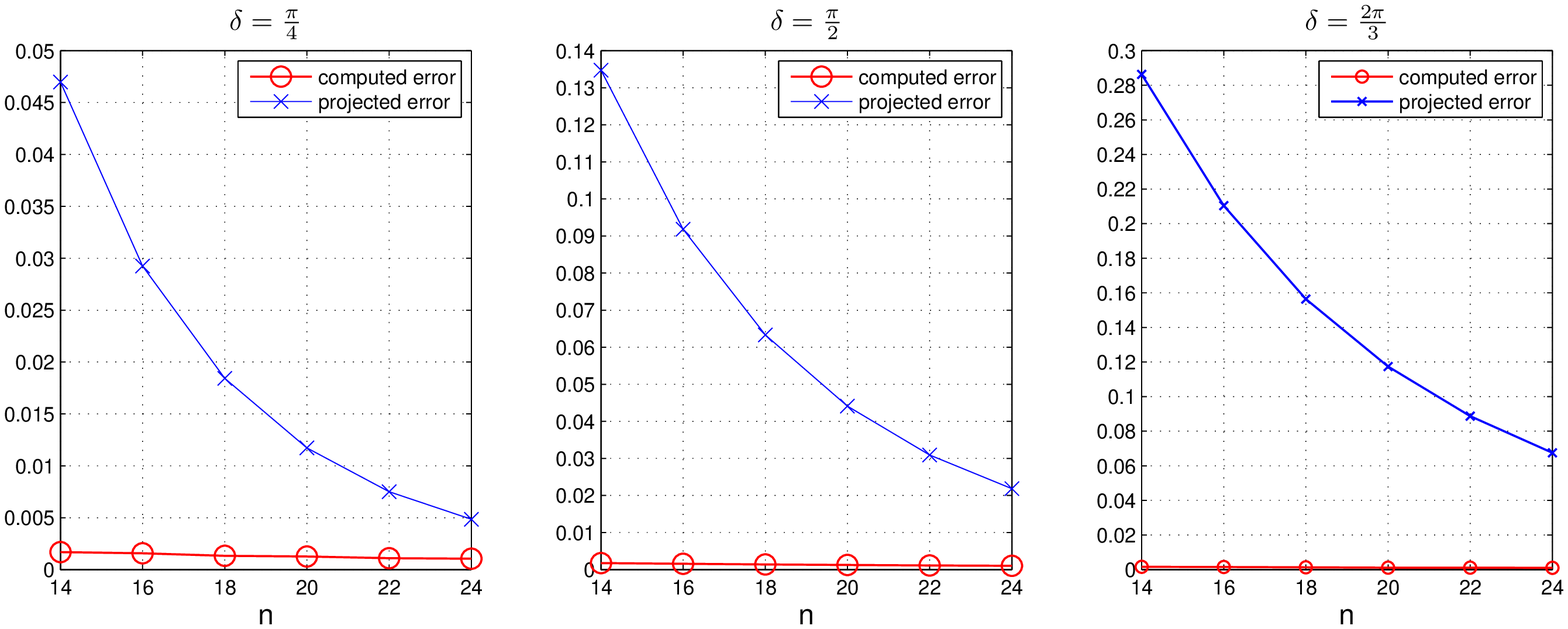}}
\end{center}

\noindent{\bf Experiment 2.} This experiment is to show that the proposed reconstruction method can converge very fast. The averaging sampling takes the form
$$
\mu_j(f):=\frac1{8}f(j-\frac1{4})+\frac34f(j)+\frac1{8}f(j+\frac1{4}),\ \ j\in\bZ,\ f\in\cB_\delta.
$$
We compute the approximation error (\ref{computederr}) and the projected error in the estimate (\ref{ultimateestimate}) for $\delta\in\{\frac{\pi}3,\frac{\pi}2,\frac{2\pi}3\}$ and for $n\in\{2,4,6,8,10,12\}$. The results are tabulated below.

$$
\begin{array}{|c|c|c|c|c|c|c|}\hline
          & n=2         &n=4  &n=6& n=8 &n=10&n=12\\\hline
\delta=\frac\pi 4&5.709\times{10^{-4}}& 2.239\times{10^{-4}}&8.689\times{10^{-5}}&2.921\times{10^{-5}}&1.976\times{10^{-5}}&1.250\times{10^{-5}}\\\hline
\delta=\frac\pi 2&1.412\times{10^{-3}}& 2.161\times{10^{-4}}&6.712\times{10^{-5}}&2.881\times{10^{-5}}&3.259\times{10^{-5}}&1.894\times{10^{-5}}\\\hline
\delta=\frac{2\pi} 3&4.023\times{10^{-4}}& 6.870\times{10^{-4}}&6.377\times{10^{-5}}&1.884\times{10^{-5}}&5.374\times{10^{-5}}&8.253\times{10^{-6}}\\\hline
\end{array}
$$

We remark that in both experiments, the computed approximation error of the proposed reconstruction method decays much faster than the upper bound in the theoretical estimate (\ref{ultimateestimate}). This is due to the reason that there might be many cancelations in adding up $f(j)\phi(x-j)$. In our estimate, we use the Cauchy-Schwartz inequality and thus view them as having the same sign.

%\section{Algorithm and Numerical Experiments}
%\setcounter{equation}{0}
\end{document}